\documentclass[11pt]{article}

\usepackage[english]{babel}

\usepackage[margin=1in]{geometry}

\usepackage{graphicx}
\usepackage{amsfonts,amsmath,amssymb,amsthm} 
\usepackage{comment}
\usepackage{xcolor}
\definecolor{ForestGreen}{rgb}{0.1333,0.5451,0.1333}
\definecolor{DarkRed}{rgb}{0.80,0,0}
\definecolor{Red}{rgb}{1,0,0}
\usepackage[linktocpage=true,
pagebackref=true,colorlinks,
linkcolor=DarkRed,citecolor=ForestGreen,
bookmarks,bookmarksopen,bookmarksnumbered]
{hyperref}
\usepackage{cleveref}
\usepackage{algorithm}
\usepackage[noend]{algpseudocode}
\usepackage{booktabs}
\usepackage{multirow}

\theoremstyle{plain}
\newtheorem{theorem}{Theorem}[section]
\newtheorem{lemma}[theorem]{Lemma}

\newtheorem{corollary}[theorem]{Corollary}

\theoremstyle{definition}
\newtheorem{definition}[theorem]{Definition}

\newcommand{\R}{\mathbb{R}}
\newcommand{\Z}{\mathbb{Z}}

\newcommand{\calM}{\mathcal{M}}

\newcommand{\bdv}{\boldsymbol{v}}

\floatname{algorithm}{Mechanism}

\title{Fair and Almost Truthful Mechanisms for Additive Valuations and Beyond}

\author{
    Biaoshuai Tao \\
    Shanghai Jiao Tong University \\
    \texttt{bstao@sjtu.edu.cn} \\
    \and
    Mingwei Yang \\
    Stanford University \\
    \texttt{mwyang@stanford.edu} \\
}
\date{}

\begin{document}
\maketitle
\date{}

\begin{abstract}
    We study the problem of fairly allocating indivisible goods among $n$ strategic agents.
    It is well-known that truthfulness is incompatible with any meaningful fairness notions.
    We bypass the strong negative result by considering the concept of \emph{incentive ratio}, a relaxation of truthfulness quantifying agents' incentive to misreport.
    Previous studies show that Round-Robin, which satisfies \emph{envy-freeness up to one good (EF1)}, achieves an incentive ratio of $2$ for additive valuations.

    In this paper, we explore the incentive ratio achievable by fair mechanisms for various classes of valuations besides additive ones.
    We first show that, for arbitrary $\epsilon > 0$, every $(\frac{1}{2} + \epsilon)$-EF1 mechanism for additive valuations admits an incentive ratio of at least $1.5$.
    Then, using the above lower bound for additive valuations in a black-box manner, we show that for arbitrary $\epsilon > 0$, every $\epsilon$-EF1 mechanism for \emph{cancelable} valuations admits an infinite incentive ratio.
    Moreover, for \emph{subadditive cancelable} valuations, we show that Round-Robin, which satisfies EF1, achieves an incentive ratio of $2$, and every $(\varphi - 1)$-EF1 mechanism admits an incentive ratio of at least $\varphi$ with $\varphi = (1 + \sqrt{5}) / 2 \approx 1.618$.
    Finally, for \emph{submodular} valuations, we show that Round-Robin, which satisfies $\frac{1}{2}$-EF1, admits an incentive ratio of $n$.

\end{abstract}
\section{Introduction}
\label{sect:intro}

In recent years, the field of discrete fair division has received extensive attention.
In the canonical model, $m$ indivisible goods, which are positively valued, are divided among a group of $n$ competing agents in a fair manner without disposal.
Arguably, the most appealing fairness notion is \emph{envy-freeness (EF)}, which is defined as each agent weakly preferring his own bundle to any other agent's bundle.
However, in the indivisible regime, EF allocations may not exist, even approximately: Consider the instance with two agents competing for one good.
Given the intractability of EF, some of its natural relaxations are \emph{envy-freeness up to one good (EF1)} proposed by \cite{DBLP:conf/sigecom/LiptonMMS04} and \cite{DBLP:conf/bqgt/Budish10}, and \emph{envy-freeness up to any good (EFX)} introduced by \cite{DBLP:journals/teco/CaragiannisKMPS19} and \cite{DBLP:conf/ecai/GourvesMT14}.
In an EF1 (EFX) allocation, agent $i$ may envy agent $j$, but the envy can be eliminated by removing one (any) good from agent $j$'s bundle.
It is known that EF1 allocations always exist and can be computed in polynomial time for general valuations \cite{DBLP:conf/sigecom/LiptonMMS04}, whereas the existence of EFX allocations remains open even for additive valuations.
See the surveys by \cite{DBLP:conf/ijcai/AmanatidisBFV22} and by \cite{DBLP:journals/sigecom/AzizLMW22} for more details of the non-strategic setting\footnote{These two surveys are later combined into one \cite{Amanatidis_2023}.}.

Nevertheless, in most practical situations, agents have no incentive to faithfully report their preferences if misreporting leads to a better outcome from their perspective.
This results in a game-theoretic perspective of the fair division problem, under which the quest for fairness becomes significantly less tractable.
Given that malicious behaviors of agents can result in severe fairness and welfare losses, a large body of research papers focus on designing mechanisms that are both truthful and fair \cite{DBLP:conf/sigecom/LiptonMMS04,DBLP:conf/aldt/CaragiannisKKK09,DBLP:conf/wine/MarkakisP11,DBLP:conf/ijcai/AmanatidisBM16}.
Unfortunately, \cite{DBLP:conf/sigecom/AmanatidisBCM17} shows that truthfulness is incompatible with any reasonable fairness concept if monetary transfers are prohibited, and this even holds for two agents with additive valuations.

Given the impossibility of combining fairness and truthfulness, the next direction to pursue is devising sufficiently fair mechanisms, which serves as the primary objective in our setting, while remaining close to truthfulness.
Even though agents are still incentivized to manipulate under the relaxed requirement for truthfulness, there still exist reasons to settle for slight relaxation.
Firstly, we measure agents' incentive to manipulate, which will be specified later, in a worst-case sense, and thus, the incentive for misreporting is very likely to be even smaller in the concerned applications.
Moreover, performing effective manipulations is costly since it requires knowing other agents' preferences, which are usually difficult to acquire, and the best response turns out to be computationally intractable even for several elementary mechanisms including Round-Robin and cut-and-choose \cite{DBLP:conf/aaai/AzizBLM17,DBLP:conf/wine/AmanatidisBFLLR21}.

As a natural relaxation of truthfulness, the notion of \emph{incentive ratio} has been widely studied under the contexts of Fisher markets \cite{DBLP:journals/iandc/ChenDTZZ22}, resource sharing \cite{DBLP:conf/ipps/ChengDL20,DBLP:conf/sigecom/ChengDLY22,DBLP:journals/dam/ChenCDQY19}, housing markets \cite{todo2019analysis}, and resource allocation \cite{DBLP:journals/jcss/HuangWWZ24,DBLP:conf/aaai/XiaoL20,DBLP:journals/corr/abs-2308-08903}.
Informally, the incentive ratio of a mechanism is defined as the worst-case ratio between the utility that an agent gains by manipulation and his utility under truthful telling.
The definition of incentive ratio is also closely related to the popular notion of \emph{approximate Nash equilibrium} \cite{DBLP:journals/sigecom/Rubinstein17,DBLP:journals/teco/CaragiannisFGS15}.

In this paper, we explore the possibility of simultaneously achieving fairness and a small incentive ratio\footnote{This resembles the concept of \emph{price of fairness} \cite{DBLP:journals/mst/BeiLMS21,DBLP:conf/wine/BarmanB020}, which captures the efficiency loss in fair allocations as opposed to the incentive loss.}.
In the most well-studied setting of additive valuations, \cite{DBLP:conf/aaai/XiaoL20} shows that Round-Robin, which satisfies EF1, admits an incentive ratio of $2$.
Inspired by the recent focus on fair division for more general valuations in both the non-strategic setting \cite{DBLP:conf/aaai/ChaudhuryGM21,DBLP:conf/aaai/BergerCFF22,DBLP:journals/teco/BarmanK20} and the strategic setting \cite{DBLP:conf/sigecom/AmanatidisBL0R23}, we also consider valuation classes that largely generalize additivity.
In particular, one of our main focuses is the class of \emph{cancelable} valuations, which generalizes budget-additive, unit-demand, and multiplicative valuations \cite{DBLP:conf/aaai/BergerCFF22} and has found its applications in various fair division results \cite{DBLP:conf/aaai/BergerCFF22,DBLP:conf/sigecom/AmanatidisBL0R23,DBLP:conf/sigecom/AkramiACGMM23}.
In addition, we also look into \emph{subadditive} and \emph{submodular} valuations, which constitute fundamental properties in combinatorial optimization and have been of recent interest in the fair division literature \cite{DBLP:conf/sigecom/AmanatidisBL0R23,DBLP:conf/aaai/ChaudhuryGM21,DBLP:journals/teco/BarmanK20}.

\subsection{Our Contributions}

In this paper, we study the incentive ratio achievable by fair mechanisms and give several positive and negative results for various classes of valuations, which are summarized in \Cref{tab:results}.
In particular, we provide the first incentive ratio lower bound strictly larger than $1$ by a constant as well as the first bounded incentive ratio upper bound beyond additive valuations.
In more detail, we describe our main contributions as follows:
\begin{itemize}
    \item For additive valuations, we show that every $(\frac{1}{2} + \epsilon)$-EF1 mechanism admits an incentive ratio of at least $1.5$, where $\epsilon > 0$ can arbitrarily depend on $n$ and $m$ (\Cref{thm:lb-12-ef1-addi}).
    This result largely improves the incentive ratio lower bound of strictly larger than $1$ by \cite{DBLP:conf/sigecom/AmanatidisBL0R23} and rules out the possibility of achieving $(\frac{1}{2} + \epsilon)$-EF1 together with an incentive ratio of arbitrarily close to $1$.

    \item For cancelable valuations, we show that every $\epsilon$-EF1 mechanism admits an infinite incentive ratio, where $\epsilon > 0$ can arbitrarily depend on $n$ and $m$ (\Cref{thm:lb-cancel-ef1}).
    In particular, our proof utilizes our incentive ratio lower bound for approximately EF1 mechanisms with additive valuations in a \emph{black-box} manner, and our result holds even for \emph{multiplicative} valuations, which constitute a subclass of cancelable valuations.

    \item We show that the impossibility result for cancelable valuations can be bypassed by the additional property of subadditivity.
    Specifically, for subadditive cancelable valuations, we show that Round-Robin, which is known to satisfy EF1 \cite{DBLP:conf/sigecom/AmanatidisBL0R23}, admits an incentive ratio of $2$ (\Cref{thm:incentive-ratio-of-round-robin-for-subadditive-cancelable-valuations}), thereby proving a separation between the cancelable and subadditive cancelable cases.
    On the negative side, we show that every $(\varphi - 1)$-EF1 mechanism for subadditive cancelable valuations admits an incentive ratio of at least $\varphi$ with $\varphi = (1 + \sqrt{5}) / 2 \approx 1.618$ (\Cref{thm:lb-ef1-subaddi-cancel}), improving the lower bound of $1.5$ given under the additive case.

    \item For submodular valuations, we show that a generalization of Round-Robin, which is proved to satisfy $\frac{1}{2}$-EF1 \cite{DBLP:conf/sigecom/AmanatidisBL0R23}, admits an incentive ratio of $n$ (\Cref{thm:incentive-ratio-of-round-robin-for-submodular-valuations}).
\end{itemize}

\begin{table}[t]
\centering
\begin{tabular}{@{}lllll@{}}
\toprule
Valuations             &  & Fairness &  & Incentive Ratio \\ \midrule
Additive               &  & EF1      &  & $[1.5, 2]$    \\
Subadditive Cancelable             &  & EF1      &  & $[\varphi, 2]$         \\
Cancelable &  & $\epsilon$-EF1      &  & $\infty$    \\
Submodular             &  & $1/2$-EF1  &  & $(1, n]$        \\ \bottomrule
\end{tabular}
\caption{Results for the incentive ratio achievable by mechanisms satisfying a specified fairness criterion for a certain class of valuations, where $\varphi = (1 + \sqrt{5}) / 2 \approx 1.618$ and $\epsilon > 0$ denotes a small real number arbitrarily depending on $n$ and $m$. The upper bound of $2$ for additive valuations is shown by \cite{DBLP:conf/aaai/XiaoL20}, and the lower bound of strictly larger than $1$ for submodular valuations is implied by the negative result of \cite{DBLP:conf/sigecom/AmanatidisBCM17} for additive valuations. The remaining bounds in the table are proved in this paper. In particular, the lower bound of $1.5$ for additive valuations holds for $(\frac{1}{2} + \epsilon)$-EF1, and the lower bound of $\varphi$ for subadditive cancelable valuations holds for $(\varphi - 1)$-EF1.}
\label{tab:results}
\end{table}

Finally, although Round-Robin is known to be prominent with additive valuations in both the non-strategic and strategic settings (see \Cref{sec:related-work}), its properties for more general valuations are less explored \cite{DBLP:conf/ecai/BouveretL14,DBLP:conf/sigecom/AmanatidisBL0R23}.
As a by-product, our positive results, which are all established via Round-Robin, further characterize the incentive guarantees of Round-Robin beyond additive valuations.

\subsection{Related Work}
\label{sec:related-work}

The mechanism design aspect of fair division is also extensively studied when resources are \emph{divisible}, which is out of the scope of this paper, and we refer to the recent paper \cite{DBLP:journals/ai/BuST23} and the references therein for a more comprehensive overview.
It is worth mentioning that in divisible resource allocations, the \emph{Maximum Nash Welfare (MNW)} and \emph{Probabilistic Serial (PS)} rules, which satisfy multiple fairness and efficiency properties, are shown to admit an incentive ratio of $2$ \cite{DBLP:journals/iandc/ChenDTZZ22,DBLP:conf/atal/0037SX24,DBLP:journals/corr/abs-2308-08903,DBLP:journals/jcss/HuangWWZ24}.
Hence, by \emph{implementing} the fractional allocations induced by the MNW or PS rules over certain integral fair allocations, randomized fair mechanisms for indivisible resources are obtained, which satisfy desirable \emph{ex-ante} incentive ratio guarantees promised by the fractional allocations \cite{DBLP:conf/sigecom/Freeman0V20,DBLP:conf/wine/Aziz20}.

Besides incentive ratio, various paradigms for bypassing the strong impossibility of combining truthfulness and fairness are also proposed.
Several recent papers \cite{DBLP:conf/wine/AmanatidisBFLLR21,DBLP:conf/sigecom/AmanatidisBL0R23} initiate the study of \emph{equilibrium fairness}, which explores the fairness guarantees of the allocations induced by \emph{pure Nash equilibria (PNE)} with respect to the underlying true valuations.
\cite{DBLP:conf/wine/AmanatidisBFLLR21} shows that Round-Robin achieves desirable equilibrium fairness properties for additive valuations.
Later on, \cite{DBLP:conf/sigecom/AmanatidisBL0R23} generalizes the equilibrium fairness guarantees of Round-Robin to cancelable and submodular valuations.
In addition, other relaxed notions of truthfulness are proposed, including the \emph{ex-ante truthfulness}~\cite{mossel2010truthful,bu2024truthfulenvyfreemechanismallocating}, \emph{maximin strategy-proofness}~\cite{brams2006better}, \emph{non-obvious manipulability}~\cite{DBLP:conf/nips/0001V22,troyan2020obvious,ortega2022obvious}, and \emph{risk-averse truthfulness}~\cite{DBLP:journals/ai/BuST23}.
Finally, another series of research considers the restricted category of dichotomous valuations and aims to design truthful mechanisms accompanied by desirable fairness and efficiency properties \cite{DBLP:conf/aaai/BabaioffEF21,DBLP:conf/wine/0002PP020,DBLP:conf/aaai/BarmanV22}.

Apart from its desirable incentive ratio and equilibrium fairness guarantees mentioned previously, Round-Robin appears as an essential tool for various fair division problems with additive valuations.
Without strategic agents, its variants are applied to produce approximate \emph{maximum share fair} allocations \cite{DBLP:journals/talg/AmanatidisMNS17,DBLP:journals/teco/BarmanK20}, EF1 allocations for \emph{mixed goods and chores} (i.e., items with negative values) \cite{DBLP:journals/aamas/AzizCIW22}, and more.
In the strategic setting, \cite{DBLP:conf/nips/0001V22} shows that Round-Robin is not obviously manipulable, and \cite{DBLP:journals/corr/abs-2306-02040} establishes that a variant of Round-Robin is \emph{Bayesian incentive compatible} when agents’ priors satisfy a neutrality condition.
\section{Preliminaries}

As conventions, given a mapping $f:X \to Y$, let $f^{-1}(y) = \{x \in X \mid f(x) = y\}$ for every $y \in Y$ and $f(X') = \{f(x) \mid x \in X'\}$ for every $X' \subseteq X$.

Let $G = \{g_1, \ldots, g_m\}$ denote the set of $m$ goods and $N = [n]$ be the set of $n$ agents.
A \emph{bundle} is a subset of $G$.
An \emph{allocation} $A = (A_1, \ldots, A_n)$ is defined as a partition of $G$ satisfying $A_i \cap A_j = \emptyset$ for all $i \neq j$ and $\bigcup_{i=1}^n A_i = G$, where $A_i$ denotes the bundle received by agent $i$. 

We assume that each agent $i$ is associated with a non-negative valuation $v_i(G')$ for each set of goods $G' \subseteq G$; for convenience, we write $v_i(g)$, $v_i(S - g)$ and $v_i(S + g)$ instead of $v_i(\{g\})$, $v_i(S \setminus \{g\})$ and $v_i(S \cup \{g\})$.
We assume that every $v_i$ is \emph{normalized}, i.e., $v_i(\emptyset) = 0$, and \emph{monotone}, i.e., $v_i(S) \leq v_i(T)$ for all $S \subseteq T \subseteq G$.
We also adopt the shortcut $v_i(T \mid S)$ for the \emph{marginal value} of a set of goods $T$ with respect to a set of goods $S$, i.e., $v_i(T \mid S) = v_i(T \cup S) - v_i(S)$.
For each agent $i$, we say that $v_i$ is
\begin{itemize}
    \item \emph{subadditive}, if $v_i(S \cup T) \leq v_i(S) + v_i(T)$ for all $S, T \subseteq G$.
    \item \emph{submodular}, if $v_i(g \mid S) \geq v_i(g \mid T)$ for all $S \subseteq T \subseteq G$ and $g \in G \setminus T$.
    \item \emph{cancelable}, if $v_i(S + g) > v_i(T + g) \implies v(S) > v(T)$ for all $S, T \subseteq G$ and $g \in G \setminus (S \cup T)$.
    \item \emph{additive}, if $v_i(S \cup T) = v_i(S) + v_i(T)$ for all $S, T \subseteq G$ with $S \cap T = \emptyset$.
\end{itemize}
Note that although both submodular and (subadditive) cancelable valuations are strict superclasses of additive valuations, neither one is a superclass of the other \cite{DBLP:conf/sigecom/AmanatidisBL0R23}.
Given an allocation $A$, define the \emph{utility} of agent $i$ as $v_i(A_i)$.

We define the fairness notion considered in this paper as follows.

\begin{definition}[$\alpha$-EF1]
    An allocation $A = (A_1, \ldots, A_n)$ is said to satisfy \emph{$\alpha$-envy-freeness up to one good ($\alpha$-EF1)} for $\alpha \in [0, 1]$ if for all $i, j \in N$, either $A_j = \emptyset$ or there exists $g \in A_j$ such that $v_i(A_i) \geq \alpha \cdot v_i(A_j - g)$.
    If $A$ satisfies $1$-EF1, we simply say that $A$ satisfies EF1.
\end{definition}

A \emph{mechanism} $\calM$ takes a valuation profile $\bdv = (v_1, \ldots, v_n)$ as input, and outputs an allocation $\calM(\bdv) = (\calM_1(\bdv), \ldots, \calM_n(\bdv))$, where $\calM_i(\bdv)$ denotes the bundle received by agent $i$.
Each agent has an underlying true valuation and is required to report a (possibly fake) valuation to the mechanism.
We adopt the notion of \emph{incentive ratio} to quantify the degree of untruthfulness of a mechanism.

\begin{definition}[Incentive Ratio]
    The \emph{incentive ratio} of a mechanism $\calM$ is defined as
    \begin{align*}
        \sup_{n,m} \sup_{v_1, \ldots, v_n} \sup_{i \in [n]} \sup_{v_i'} \frac{v_i(\calM_i(v_1, \ldots, v_i', \ldots, v_n))}{v_i(\calM_i(v_1, \ldots, v_i, \ldots, v_n))}.
    \end{align*}
\end{definition}

Observe that the incentive ratio of every mechanism is at least $1$ by setting $v_i' = v_i$.
If the incentive ratio of a mechanism is $1$, then we say that the mechanism is \emph{truthful}.

\subsection{Strongly Desire and Control}

Recall that \cite{DBLP:conf/sigecom/AmanatidisBCM17} proposes the notions of \emph{strongly desire} and \emph{control} in the context of truthfulness with two agents.
We generalize these notions and accommodate them to the concept of incentive ratio.
In this subsection, we assume that there are $n = 2$ agents and do not make any restrictions on valuations except for the default that they are normalized and monotone.

For $\alpha \geq 1$, we say that an agent $i$ \emph{$\alpha$-strongly desires} a good $g$ if he values $g$ strictly more than all goods in $G \setminus \{g\}$ combined multiplying by $\alpha$, i.e., $v_i(g) > \alpha \cdot v_i(G - g)$.
Next, we define the notion of $\alpha$-control.

\begin{definition}[$\alpha$-Control]
    Given a mechanism $\calM$ and $\alpha \geq 1$, we say that an agent $i$ \emph{$\alpha$-controls} a good $g$ with respect to $\calM$ if for every profile $\bdv$ where agent $i$ $\alpha$-strongly desires $g$, $g \in \calM_i(\bdv)$.
\end{definition}

Given a mechanism $\calM$ and $\alpha \geq 1$, every good $g$ is $\alpha$-controlled by at most one agent with respect to $\calM$ since when both agents $\alpha$-strongly desire $g$, only one agent can receive it.
Moreover, assuming that $\calM$ admits an incentive ratio of $\alpha$, we show in the following lemma that every good is $\alpha$-controlled by exactly one agent with respect to $\calM$.

\begin{lemma}\label{lem:control-lemma}
    Given a mechanism $\calM$ with an incentive ratio of $\alpha \geq 1$, every $g \in G$ is $\alpha$-controlled by exactly one agent with respect to $\calM$.
\end{lemma}

\begin{proof}
    Let $\bdv = (v_1, v_2)$ be a profile where both agents $\alpha$-strongly desire $g$.
    Assume without loss of generality that $g \in \calM_1(\bdv)$, and we show that $g$ is $\alpha$-controlled by agent $1$ with respect to $\calM$.
    Let $\bdv' = (v_1', v_2')$ be an arbitrary profile in which agent $1$ $\alpha$-strongly desires $g$, and we aim to show that $g \in \calM_1(\bdv')$.
    Initially, consider the intermediate profile $\bdv^* = (v_1, v_2')$.
    If $g \in \calM_2(\bdv^*)$, then agent $2$ would deviate from $\bdv$ to $\bdv^*$ to improve his utility in $\bdv$ by strictly more than $\alpha$ times.
    Hence, by the incentive ratio $\alpha$ of $\calM$, $g \in \calM_1(v^*)$.
    Similarly, if $g \in \calM_2(\bdv')$, then agent $1$ would deviate from $\bdv'$ to $\bdv^*$ to improve his utility in $\bdv'$ by strictly more than $\alpha$ times.
    Hence, by the incentive ratio $\alpha$ of $\calM$, $g \in \calM_1(\bdv')$, concluding that agent $1$ $\alpha$-controls $g$ with respect to $\calM$.
\end{proof}

\section{Additive Valuations}

In this section, we consider additive valuations and show an incentive ratio lower bound of $1.5$ for $(\frac{1}{2} + \epsilon)$-EF1 mechanisms, where $\epsilon > 0$ can arbitrarily depend on $n$ and $m$.

\begin{theorem}\label{thm:lb-12-ef1-addi}
    Every $(\frac{1}{2} + \epsilon)$-EF1 mechanism for additive valuations admits an incentive ratio of at least $1.5$, where $\epsilon > 0$ can arbitrarily depend on $n$ and $m$.
\end{theorem}


The rest of this section is devoted to proving \Cref{thm:lb-12-ef1-addi}, for which we construct a series of profiles and show that $(\frac{1}{2} + \epsilon)$-EF1 and an incentive ratio of strictly smaller than $1.5$ cannot be simultaneously guaranteed in all these profiles.
Assume for contradiction that a $(\frac{1}{2} + \epsilon)$-EF1 mechanism $\calM$ for additive valuations exists with an incentive ratio of $\alpha$ satisfying $1 \leq \alpha < 1.5$.
Suppose that there are $n = 2$ additive agents and $m = 7$ goods.
For every $i \in [2]$, denote $N_i$ as the set of goods $\alpha$-controlled by agent $i$ with respect to $\calM$.
By \Cref{lem:control-lemma}, $(N_1, N_2)$ forms a partition of $G$.
Without loss of generality, assume that $|N_1| \geq 4$ and $\{g_1, g_2, g_3, g_4\} \subseteq N_1$.
Denote $G' = \{g_1, g_2, g_3, g_4\}$, and every constructed additive valuation $v$ in the proof will satisfy $v(G \setminus G') = 0$.
For simplicity, we assume goods in $G \setminus G'$ always to be assigned to agent $1$, and we omit them when describing valuations and allocations.

We start with the profile $\bdv^{(0)} = (v_1, v_2)$ where

\begin{table}[H]
\centering
\begin{tabular}{@{}cccccc@{}}
\toprule
\multicolumn{1}{l}{}          &       & $g_1$ & $g_2$ & $g_3$ & $g_4$ \\ \midrule
\multirow{2}{*}{$\bdv^{(0)}$} & $v_1$ & $1$   & $1$   & $1$   & $1$   \\
                              & $v_2$ & $1$   & $1$   & $1$   & $1$   \\ \bottomrule
\end{tabular}
\end{table}

\noindent By the $(\frac{1}{2} + \epsilon)$-EF1 property of $\calM$, $|\calM_1(\bdv^{(0)})| = |\calM_2(\bdv^{(0)})| = 2$.
Without loss of generality, assume that $\calM_1(\bdv^{(0)}) = \{g_1, g_2\}$ and $\calM_2(\bdv^{(0)}) = \{g_3, g_4\}$.

Let $\delta$ be an arbitrary real number with $\delta \geq 5$, and we consider the next profile $\bdv^{(1)} = (v_1', v_2)$ where

\begin{table}[H]
\centering
\begin{tabular}{@{}cccccc@{}}
\toprule
\multicolumn{1}{l}{}          &       & $g_1$ & $g_2$ & $g_3$ & $g_4$ \\ \midrule
\multirow{2}{*}{$\bdv^{(1)}$} & $v_1'$ & $\delta$   & $1.5 \delta$   & $0$   & $0$   \\
                              & $v_2$ & $1$   & $1$   & $1$   & $1$   \\ \bottomrule
\end{tabular}
\end{table}

\noindent We claim that $\calM_1(\bdv^{(1)}) = \{g_1, g_2\}$ and $\calM_2(\bdv^{(1)}) = \{g_3, g_4\}$.
Firstly, by the $(\frac{1}{2} + \epsilon)$-EF1 property of $\calM$, $|\calM_2(\bdv^{(1)})| \geq 2$.
Moreover, if $|\{g_1, g_2\} \cap \calM_1(\bdv^{(1)})| < 2$, by deviating from $\bdv^{(1)}$ to $\bdv^{(0)}$, agent $1$ can increase his utility in $\bdv^{(1)}$ by a factor of
\begin{align*}
    \frac{v_1'(\calM_1(\bdv^{(0)}))}{v_1'(\calM_1(\bdv^{(1)}))}
    \geq \frac{v_1'(\{g_1, g_2\})}{v_1'(g_2)}
    = \frac{2.5}{1.5}
    > \alpha,
\end{align*}
violating the incentive ratio $\alpha$ of $\calM$.
Hence, $\{g_1, g_2\} \subseteq \calM_1(\bdv^{(1)})$, and it follows that $\calM_1(\bdv^{(1)}) = \{g_1, g_2\}$ and $\calM_2(\bdv^{(1)}) = \{g_3, g_4\}$.

We proceed to the next profile $\bdv^{(2)} = (v_1'', v_2)$ where

\begin{table}[H]
\centering
\begin{tabular}{@{}cccccc@{}}
\toprule
\multicolumn{1}{l}{}          &       & $g_1$ & $g_2$ & $g_3$ & $g_4$ \\ \midrule
\multirow{2}{*}{$\bdv^{(2)}$} & $v_1''$ & $1.5 \delta$   & $\delta$   & $0$   & $0$   \\
                              & $v_2$ & $1$   & $1$   & $1$   & $1$   \\ \bottomrule
\end{tabular}
\end{table}

\noindent Analogous to $\bdv^{(1)}$, we can show that $\calM_1(\bdv^{(2)}) = \{g_1, g_2\}$ and $\calM_2(\bdv^{(2)}) = \{g_3, g_4\}$.

In the next profile, we modify the valuation of agent $2$.
Define $\bdv^{(3)} = (v_1', v_2')$ where

\begin{table}[H]
\centering
\begin{tabular}{@{}cccccc@{}}
\toprule
\multicolumn{1}{l}{}          &       & $g_1$ & $g_2$ & $g_3$ & $g_4$ \\ \midrule
\multirow{2}{*}{$\bdv^{(3)}$} & $v_1'$ & $\delta$   & $1.5 \delta$   & $0$   & $0$   \\
                              & $v_2'$ & $0$   & $\delta$   & $1$   & $1$   \\ \bottomrule
\end{tabular}
\end{table}

\noindent We claim that $\calM_1(\bdv^{(3)}) = \{g_1, g_2\}$ and $\calM_2(\bdv^{(3)}) = \{g_3, g_4\}$.
Firstly, since agent $1$ $\alpha$-strongly desires $g_2$ in $\bdv^{(3)}$, $g_2 \in \calM_1(\bdv^{(3)})$ by the assumption that agent $1$ $\alpha$-controls $g_2$ with respect to $\calM$.
Moreover, if $|\{g_3, g_4\} \cap \calM_2(\bdv^{(3)})| < 2$, by deviating from $\bdv^{(3)}$ to $\bdv^{(1)}$, agent $2$ can increase his utility in $\bdv^{(3)}$ by a factor of
\begin{align*}
    \frac{v_2'(\calM_2(\bdv^{(1)}))}{v_2'(\calM_2(\bdv^{(3)}))}
    \geq 2
    > \alpha,
\end{align*}
violating the incentive ratio $\alpha$ of $\calM$.
Hence, $\{g_3, g_4\} \subseteq \calM_2(\bdv^{(3)})$.
Finally, if $|\calM_2(\bdv^{(3)})| > 2$, by deviating from $\bdv^{(1)}$ to $\bdv^{(3)}$, agent $2$ can increase his utility in $\bdv^{(1)}$ by a factor of
\begin{align*}
    \frac{v_2(\calM_2(\bdv^{(3)}))}{v_2(\calM_2(\bdv^{(1)}))}
    \geq \frac{3}{2}
    > \alpha,
\end{align*}
violating the incentive ratio $\alpha$ of $\calM$.
As a result, $|\calM_2(\bdv^{(3)})| \leq 2$, and it follows that $\calM_1(\bdv^{(3)}) = \{g_1, g_2\}$ and $\calM_2(\bdv^{(3)}) = \{g_3, g_4\}$.

In the next profile, we manage to allocate $g_2$ to agent $2$.
Define $\bdv^{(4)} = (v_1'', v_2'')$ where

\begin{table}[H]
\centering
\begin{tabular}{@{}cccccc@{}}
\toprule
\multicolumn{1}{l}{}          &       & $g_1$ & $g_2$ & $g_3$ & $g_4$ \\ \midrule
\multirow{2}{*}{$\bdv^{(4)}$} & $v_1''$ & $1.5 \delta$   & $\delta$   & $0$   & $0$   \\
                              & $v_2''$ & $\delta$   & $\delta$   & $1$   & $1$   \\ \bottomrule
\end{tabular}
\end{table}

\noindent We claim that $g_1 \in \calM_1(\bdv^{(4)})$ and $g_2 \in \calM_2(\bdv^{(4)})$.
Firstly, since agent $1$ $\alpha$-strongly desires $g_1$ in $\bdv^{(4)}$, $g_1 \in \calM_1(\bdv^{(4)})$ by the assumption that agent $1$ $\alpha$-controls $g_1$ with respect to $\calM$.
Moreover, if $\{g_1, g_2\} \subseteq \calM_1(\bdv^{(4)})$, then $\calM(\bdv^{(4)})$ is not $(\frac{1}{2} + \epsilon)$-EF1 for agent $2$, violating the $(\frac{1}{2} + \epsilon)$-EF1 property of $\calM$.
Hence, $|\{g_1, g_2\} \cap \calM_1(\bdv^{(4)})| \leq 1$, and it follows that $g_1 \in \calM_1(\bdv^{(4)})$ and $g_2 \in \calM_2(\bdv^{(4)})$.

We present our final profile to derive a contradiction.
Define $\bdv^{(5)} = (v_1'', v_2')$ where

\begin{table}[H]
\centering
\begin{tabular}{@{}cccccc@{}}
\toprule
\multicolumn{1}{l}{}          &       & $g_1$ & $g_2$ & $g_3$ & $g_4$ \\ \midrule
\multirow{2}{*}{$\bdv^{(5)}$} & $v_1''$ & $1.5 \delta$   & $\delta$   & $0$   & $0$   \\
                              & $v_2'$ & $0$   & $\delta$   & $1$   & $1$   \\ \bottomrule
\end{tabular}
\end{table}

\noindent Firstly, if $|\{g_1, g_2\} \cap \calM_1(\bdv^{(5)})| < 2$, by deviating from $\bdv^{(5)}$ to $\bdv^{(3)}$, agent $1$ can increase his utility in $\bdv^{(5)}$ by a factor of
\begin{align*}
    \frac{v_1''(\calM_1(\bdv^{(3)}))}{v_1''(\calM_1(\bdv^{(5)}))}
    \geq \frac{v_1''(\{g_1, g_2\})}{v_1''(\{g_1\})}
    = \frac{2.5}{1.5}
    > \alpha,
\end{align*}
violating the incentive ratio $\alpha$ of $\calM$.
Hence, $\{g_1, g_2\} \subseteq \calM_1(\bdv^{(5)})$.
However, by deviating from $\bdv^{(5)}$ to $\bdv^{(4)}$, agent $2$ can increase his utility in $\bdv^{(5)}$ by a factor of
\begin{align*}
    \frac{v_2'(\calM_2(\bdv^{(4)}))}{v_2'(\calM_2(\bdv^{(5)}))}
    \geq \frac{v_2'(g_2)}{v_2'(\{g_3, g_4\})}
    = \frac{\delta}{2}
    > \alpha,
\end{align*}
violating the incentive ratio $\alpha$ of $\calM$.
This concludes the proof of \Cref{thm:lb-12-ef1-addi}.

\section{Cancelable Valuations}

In this section, we consider cancelable valuations and show that every $\epsilon$-EF1 mechanism admits an infinite incentive ratio, where $\epsilon>0$ can arbitrarily depend on $n$ and $m$.

\begin{theorem}\label{thm:lb-cancel-ef1}
    Every $\epsilon$-EF1 mechanism for cancelable valuations admits an infinite incentive ratio, where $\epsilon > 0$ can arbitrarily depend on $n$ and $m$.
\end{theorem}

The rest of this section is devoted to proving \Cref{thm:lb-cancel-ef1}.
In particular, we establish \Cref{thm:lb-cancel-ef1} by showing a stronger statement that every $\epsilon$-EF1 mechanism for \emph{multiplicative} valuations, which constitute a subset of cancelable valuations \cite{DBLP:conf/aaai/BergerCFF22}, admits an infinite incentive ratio.
Recall that a valuation $v$ is multiplicative if $v(S) = \prod_{g \in S} v(g)$ for every $S \subseteq G$ with $|S| \geq 1$.
Moreover, since we assume valuations to be normalized and monotone, a multiplicative valuation $v$ should also satisfy $v(\emptyset) = 0$ and $v(g) \geq 1$ for every $g \in G$.

Assume for contradiction that there exists an $\epsilon$-EF1 mechanism $\calM^c$ for multiplicative valuations with an incentive ratio of $\alpha < \infty$, where $\alpha$ can arbitrarily depend on $n$ and $m$.
We construct a mechanism $\calM$ for additive valuations as follows, which we will show that satisfies $\frac{9}{10}$-EF1 with an incentive ratio of at most $1.1$, violating \Cref{thm:lb-12-ef1-addi}.
Given $\delta > 0$ and an additive valuation $v$, define $v^{\delta}$ as the multiplicative valuation satisfying $v^{\delta}(\emptyset) = 0$ and $v^{\delta}(S) = \exp(\delta \cdot v(S))$ for every $S \subseteq G$ with $|S| \geq 1$, and it is easy to verify that $v^{\delta}$ is normalized and monotone.
Given as input an additive valuation profile $(v_1, \ldots, v_n)$, let $\delta > 0$ be a sufficiently large real number such that for all $i \in N$ and $S \subseteq G$ with $v_i(S) > 0$,
\begin{align}\label{eqn:delta-con}
    \frac{\max\{ \ln(1/\epsilon), \ln \alpha \}}{\delta \cdot v_i(S)} \leq \frac{1}{10},
\end{align}
and $\calM$ outputs the allocation $\calM^c(v_1^{\delta}, \ldots, v_n^{\delta})$.
Note that $\delta$ is a function of $(v_1, \ldots, v_n)$, $\epsilon$, and $\alpha$.

We first show that $\calM$ satisfies $\frac{9}{10}$-EF1 for additive valuations.

\begin{lemma}\label{lem:addi-fair}
    $\calM$ satisfies $\frac{9}{10}$-EF1 for additive valuations.
\end{lemma}

\begin{proof}
    Fix an additive valuation profile $(v_1, \ldots, v_n)$, and let $A = \calM(v_1, \ldots, v_n) = \calM^c(v_1^{\delta}, \ldots, v_n^{\delta})$.
    Fix $i, j \in N$, and we show that $A$ satisfies $\frac{9}{10}$-EF1 for the pair of agents $(i, j)$ with respect to $(v_1, \ldots, v_n)$, i.e., if $A_j \neq \emptyset$, then there exists $g \in A_j$ such that $v_i(A_i) \geq \frac{9}{10} \cdot v_i(A_j - g)$.
    By the $\epsilon$-EF1 property of $\calM^c$, if $A_j \neq \emptyset$, then there exists $g \in A_j$ such that $v_i^{\delta}(A_i) \geq \epsilon \cdot v_i^{\delta}(A_j - g)$, which is equivalent to
    \begin{align}\label{eqn:ef1-multi}
        \exp(\delta \cdot v_i(A_i))
        \geq \epsilon \cdot \exp(\delta \cdot v_i(A_j - g)).
    \end{align}
    Assume that $v_i(A_j - g) > 0$, since otherwise, $\frac{9}{10}$-EF1 is straightforwardly satisfied for pair $(i, j)$ with respect to $(v_1, \ldots, v_n)$.
    Taking logarithm on both sides of \eqref{eqn:ef1-multi} and rearranging the terms, we obtain
    \begin{align*}
        v_i(A_i)
        \geq \frac{\ln \epsilon}{\delta} + v_i(A_j - g)
        \geq \frac{9}{10} \cdot v_i(A_j - g),
    \end{align*}
    where the second inequality holds by the assumption that $v_i(A_j - g) > 0$ and \eqref{eqn:delta-con}, implying that $A$ satisfies $\frac{9}{10}$-EF1 for pair $(i, j)$ with respect to $(v_1, \ldots, v_n)$.
    Therefore, $A$ satisfies $\frac{9}{10}$-EF1 with respect to $(v_1, \ldots, v_n)$, concluding that $\calM$ satisfies $\frac{9}{10}$-EF1 for additive valuations.
\end{proof}

Next, we show an incentive ratio upper bound of $1.1$ for $\calM$ with additive valuations.

\begin{lemma}\label{lem:addi-ir}
    $\calM$ admits an incentive ratio of at most $1.1$ for additive valuations.
\end{lemma}

\begin{proof}
    Due to symmetry, it is sufficient to show that agent $1$ cannot increase his utility under $\calM$ by a factor of strictly larger than $1.1$ via misreporting.
    Fix an additive valuation profile $(v_1, \ldots, v_n)$, and suppose that agent $1$ manipulates his valuation as $\hat{v}_1$.
    By the definition of $\calM$, the utilities of agent $1$ with and without manipulation are $v_1(\calM_1(\hat{v}_1, v_2, \ldots, v_n)) = v_1(\calM_1^c(\hat{v}_1^{\delta}, v_2^{\delta}, \ldots, v_n^{\delta}))$ and $v_1(\calM_1(v_1, \ldots, v_n)) = v_1(\calM_1^c(v_1^{\delta}, \ldots, v_n^{\delta}))$, respectively.
    By the incentive ratio $\alpha$ of $\calM^c$,
    \begin{align}\label{eqn:ic-multi}
        &v_1^{\delta}(\calM^c_1(\hat{v}_1^{\delta}, v_2^{\delta}, \ldots, v_n^{\delta}))
        \leq \alpha \cdot v_1^{\delta}(\calM_1^c(v_1^{\delta}, \ldots, v_n^{\delta})).
    \end{align}
    Taking logarithm on both sides and by the definition of $v_1^{\delta}$, \eqref{eqn:ic-multi} is equivalent to
    \begin{align}\label{eqn:ic-multi-modi}
        \delta \cdot v_1(\calM_1^c(\hat{v}_1^{\delta}, v_2^{\delta}, \ldots, v_n^{\delta}))
        \leq \ln \alpha + \delta \cdot v_1(\calM^c_1(v_1^{\delta}, \ldots, v_n^{\delta})).
    \end{align}
    If $v_1(\calM^c_1(v_1^{\delta}, \ldots, v_n^{\delta})) = 0$, then $v_1(\calM_1^c(\hat{v}_1^{\delta}, v_2^{\delta}, \ldots, v_n^{\delta})) = 0$ by \eqref{eqn:ic-multi-modi} and \eqref{eqn:delta-con}, which implies that agent $1$ cannot increase his utility under $\calM$ by misreporting $\hat{v}_1$.
    From now on, we assume that $v_1(\calM_1^c(v_1^{\delta}, \ldots, v_n^{\delta})) > 0$.

    Note that by dividing $\delta \cdot v_1(\calM_1^c(v_1^{\delta}, \ldots, v_n^{\delta})) > 0$ on both sides of \eqref{eqn:ic-multi-modi} and rearranging the terms, we obtain
    \begin{align}\label{eqn:misrep-mc}
        \frac{v_1(\calM_1^c(\hat{v}_1^{\delta}, v_2^{\delta}, \ldots, v_n^{\delta}))}{v_1(\calM_1^c(v_1^{\delta}, v_2^{\delta}, \ldots, v_n^{\delta}))}
        \leq \frac{\ln \alpha}{\delta \cdot v_1(\calM_1^c(v_1^{\delta}, \ldots, v_n^{\delta}))} + 1
        \leq 1.1,
    \end{align}
    where the last inequality holds by the assumption that $v_1(\calM_1^c(v_1^{\delta}, \ldots, v_n^{\delta})) > 0$ and \eqref{eqn:delta-con}.
    Hence, by misreporting $\hat{v}_1$, agent $1$ can increase his utility under $\calM$ by a factor of
    \begin{align*}
        \frac{v_1(\calM_1(\hat{v}_1, v_2, \ldots, v_n))}{v_1(\calM_1(v_1, v_2, \ldots, v_n))}
        = \frac{v_1(\calM^c_1(\hat{v}_1^{\delta}, v_2^{\delta}, \ldots, v_n^{\delta}))}{v_1(\calM^c_1(v_1^{\delta}, v_2^{\delta}, \ldots, v_n^{\delta}))}
        \leq 1.1,
    \end{align*}
    where the equality holds by the definition of $\calM$ and the inequality holds by \eqref{eqn:misrep-mc}, concluding that the incentive ratio of $\calM$ is upper bounded by $1.1$ for additive valuations.
\end{proof}

Finally, combining \Cref{lem:addi-fair} and \Cref{lem:addi-ir}, it follows that $\calM$ satisfies $\frac{9}{10}$-EF1 with an incentive ratio of at most $1.1$ for additive valuations, which contradicts \Cref{thm:lb-12-ef1-addi}, concluding the proof of \Cref{thm:lb-cancel-ef1}.

\section{Subadditive Cancelable Valuations}

We have shown in \Cref{thm:lb-cancel-ef1} that every EF1 mechanism for cancelable valuations admits an infinite incentive ratio.
In this section, we show that this impossibility result can be bypassed with the additional property of subadditivity.
In particular, for subadditive cancelable valuations, we show that Round-Robin, which satisfies EF1, admits an incentive ratio of $2$.
Then, we complement our positive result by providing an incentive ratio lower bound of $\varphi = (1 + \sqrt{5}) / 2 \approx 1.618$ for $(\varphi - 1)$-EF1 mechanisms with subadditive cancelable valuations, improving the lower bound of $1.5$ implied by \Cref{thm:lb-12-ef1-addi}.

\subsection{Upper Bound}

We first present our positive result.
Recall that Round-Robin, which is formally presented in Mechanism~\ref{alg:round-robin}, consists of multiple \emph{rounds}, and at each round, agents alternately receive an available good with the highest value.
When multiple goods have the same value, we assume agents always break ties lexicographically, i.e., breaking ties in favor of the choice with the smallest index.
We call the process that an agent receives a good at a \emph{stage}, and there are $m$ stages in total.

\begin{algorithm}[!htp]
        \caption{Round-Robin}
        \label{alg:round-robin}
        \begin{algorithmic}[1]
            \State {$S = G$; $(A_1, \ldots, A_n) = (\emptyset, \ldots, \emptyset)$; $k = \lceil m / n \rceil$}
            \For {$r = 1, \ldots, k$}
                \For {$i = 1, \ldots, n$}
                    \State {$g = \arg\max_{h \in S} v_i(h)$} \Comment{Break ties lexicographically.}
                    \State {$A_i = A_i \cup \{g\}$} \Comment{The current agent receives his favorite available good.}
                    \State {$S = S \setminus \{g\}$} \Comment{The good is no longer available.}
                \EndFor
            \EndFor
            \State {\Return $A = (A_1, \ldots, A_n)$}
        \end{algorithmic}
\end{algorithm}

For cancelable valuations, \cite{DBLP:conf/sigecom/AmanatidisBL0R23} shows that Mechanism~\ref{alg:round-robin} satisfies EF1, and hence, Mechanism~\ref{alg:round-robin} admits an infinite incentive ratio by \Cref{thm:lb-cancel-ef1}.
Nevertheless, we show that the incentive ratio of Mechanism~\ref{alg:round-robin} can be improved to $2$ with the additional property of subadditivity.

\begin{theorem}\label{thm:incentive-ratio-of-round-robin-for-subadditive-cancelable-valuations}
    Mechanism~\ref{alg:round-robin} admits an incentive ratio of $2$ for subadditive cancelable valuations.
\end{theorem}

The rest of this subsection is devoted to proving \Cref{thm:incentive-ratio-of-round-robin-for-subadditive-cancelable-valuations}.
Note that the lower bound in \Cref{thm:incentive-ratio-of-round-robin-for-subadditive-cancelable-valuations} is implied by the lower bound for the incentive ratio of Mechanism~\ref{alg:round-robin} for additive valuations \cite{DBLP:conf/aaai/XiaoL20}, and it remains to prove the upper bound.

We first present a crucial property of cancelable valuations.

\begin{lemma}[\cite{DBLP:conf/sigecom/AmanatidisBL0R23}]\label{lem:a-property-of-cancelable-valuations}
    Suppose that $v(\cdot)$ is cancelable.
    Let $X = \{x_1, \ldots, x_k\} \subseteq G$ and $Y = \{y_1, \ldots, y_k\} \subseteq G$.
    If $v(x_j) \geq v(y_j)$ for every $j \in [k]$, then $v(X) \geq v(Y)$.
\end{lemma}

Since every agent $i$ cannot alter the goods chosen in the first $i - 1$ stages by manipulation, it is sufficient to prove the incentive ratio for agent $1$.
Assuming all agents report truthfully, we renumber the goods so that for every $i \in [m]$, $g_i$ is the good received by some agent in stage $i$, i.e., $g_i$ is the favorite good among all the remaining goods for the agent who is designated to receive a good in stage $i$.
For $i \in \{0, \ldots, m\}$, define $G_i=\{g_{i + 1}, \ldots, g_m\}$ as the set of remaining goods at the end of stage $i$ and $B_i$ as the set of goods received by agent $1$ until the end of stage $i$.

Now, assume that agent $1$ manipulates his valuation, and we run Mechanism~\ref{alg:round-robin} again on the manipulated valuation profile.
For every $i \in \{0, \ldots, m\}$, let $g_i'$ be the good received by some agent in stage $i$, and define $G_i' = \{g_{i + 1}', \ldots, g_m'\}$ and $B_i'$ analogously.
A crucial observation, which will be formalized later, is that $G_i' \setminus G_i$ includes all the goods in $\{g_1, \ldots, g_i\}$ ``left'' to the subsequent stages by agent $1$ via manipulation, and in the extreme case, they will all end up being received by agent $1$.

For every $i \in \{0, \ldots, m\}$, let $X_i = B_i' \cup (G_i' \setminus G_i)$ be the set of goods possibly obtained by agent $1$ among goods in $\{g_1, \ldots, g_i\}$.
Our goal is to establish good-wise comparisons between goods in $X_i$ and $B_i$.
Specifically, we hope to assign each good in $X_i$ to a good in $B_i$ such that the value of the former with respect to agent $1$'s true valuation is upper bounded by that of the latter, and each good in $B_i$ is assigned with at most two goods in $X_i$.
In particular, we show that there exist mappings $M_i: X_i \to B_i$ for every $i \in \{0, \ldots, m\}$ satisfying the following properties:
\begin{enumerate}
    \item \label{item:property1-for-item-mappings} for every $g \in X_i$, $v_1(g) \leq v_1(M_i(g))$, and
    \item \label{item:property2-for-item-mappings} for every $g \in B_i$, $|M_i^{-1}(g)| \leq 2$.
\end{enumerate}

We demonstrate the implication of the existence of such mappings.
With the mapping $M_m: X_m \to B_m$ satisfying both properties, by Property~\ref{item:property2-for-item-mappings}, there exists a partition $(R_1, R_2)$ of $X_m = B'_m$ such that for every $g \in B_m$, $|M_m^{-1}(g) \cap R_1| \leq 1$ and $|M_m^{-1}(g) \cap R_2| \leq 1$.
Furthermore, by Property~\ref{item:property1-for-item-mappings}, we apply \Cref{lem:a-property-of-cancelable-valuations} twice with respectively $X = R_1$, $Y = M_m(R_1)$ and $X = R_2$, $Y = M_m(R_2)$ to obtain
\begin{align*}
    v_1(R_1)
    \leq v_1(M_m(R_1))
    \leq v_1(B_m)
\end{align*}
where the second inequality holds by the monotonicity of $v_1$, and similarly,
\begin{align*}
    v_1(R_2)
    \leq v_1(M_m(R_2))
    \leq v_1(B_m).
\end{align*}
As a result, by the subadditivity of $v_1$,
\begin{align*}
    v_1(B_m')
    \leq v_1(R_1) + v_1(R_2)
    \leq 2v_1(B_m).
\end{align*}
This indicates that agent $1$ cannot gain a utility of more than $2v_1(B_m)$ by manipulation, where $v_1(B_m)$ equals his utility when reporting truthfully, concluding the proof of \Cref{thm:incentive-ratio-of-round-robin-for-subadditive-cancelable-valuations}.

It remains to prove the existence of such mappings, which is provided in the following lemma.

\begin{lemma}\label{lem:exist-mapping}
    There exist mappings $M_i: X_i \to B_i$ for every $i \in \{0, \ldots, m\}$ satisfying Properties~\ref{item:property1-for-item-mappings} and~\ref{item:property2-for-item-mappings}.
\end{lemma}

\begin{proof}
    We say that $M_i$ is \emph{valid} if it satisfies Properties~\ref{item:property1-for-item-mappings} and \ref{item:property2-for-item-mappings}.
    For $i \in [m]$, let $b_i \in B_i$ be the good received by agent $1$ the latest among all goods in $B_i$ when he reports truthfully.
    Observe that $v_1(b_i) \geq \max_{g \in G_i} v_1(g)$ by the description of Mechanism~\ref{alg:round-robin}.
    For $i = 0$, $X_0 = \emptyset$ since $B_0 = B'_0 = \emptyset$ and $G_0 = G_0' = G$, and hence, a valid mapping $M_0$ straightforwardly exists.
    Assume for induction that $M_{k - 1}$ is a valid mapping with $k \in [m]$, and we show how to construct a valid mapping $M_k$ based on $M_{k - 1}$.
    For convenience, we call the goods in $X_k \setminus X_{k - 1}$ as new goods and all other goods in $X_k$ as old goods.
    We define $M_k(g') = M_{k - 1}(g')$ for each old good $g'$, and it remains to specify $M_k(g)$ for each new good $g$

    If $k = cn + 1$ for some $c \in \Z_{\geq 0}$, i.e., it is agent $1$'s turn to receive his favorite good, then we have $B_k = B_{k - 1} \cup \{g_k\}$, $B_k' = B_{k - 1}' \cup \{g_k'\}$, and $b_k = g_k$.
    Note that the only possible new goods are $g_k$ and $g_k'$.
    For each new good $g$, let $M_k(g) = g_k$.
    If $g_k = g_k'$, then it is easy to verify that $M_k$ satisfies both Properties~\ref{item:property1-for-item-mappings} and~\ref{item:property2-for-item-mappings}.
    Now, assume that $g_k \neq g_k'$, and we show that $M_k$ constructed above is a valid mapping.
    Firstly, Property~\ref{item:property2-for-item-mappings} is satisfied as only $g_k$ and $g_k'$ might be contained in $M_k^{-1}(b_k)$.
    Besides, if $g_k$ is a new good, then Property~\ref{item:property1-for-item-mappings} straightforwardly holds for $g_k$.
    Finally, if $g_k'$ is a new good, then $g_k' \notin X_{k - 1}$, which implies that $g_k' \in G_{k - 1}$, and by the definition of $g_k$, we have $v_1(g_k') \leq \max_{g \in G_{k - 1}} v_1(g) \leq v_1(g_k) = v_1(M_k(g_k'))$.
    Thus, Property~\ref{item:property1-for-item-mappings} also holds for $g_k'$.
    As a result, $M_k$ is a valid mapping.

    If $k = cn+j$ for some $c \in \Z_{\geq 0}$ and $j \in \{2, \ldots, n\}$, then we have $B_k = B_{k - 1}$ and $B_k' = B_{k - 1}'$.
    Note that the only possible new good is $g_k$.
    Assume that $g_k \neq g_k'$, as otherwise, we have $X_k = X_{k - 1}$ and we are done.
    If $g_k \notin G_k'$, then we are also done since $X_k \subseteq X_{k - 1}$.
    Now, assume that $g_k \in G_k' \subseteq G_{k - 1}'$.
    Since $g_k$ is agent $j$'s favorite good in $G_{k - 1}$ and agent $j$, who reports truthfully, receives $g_k' \neq g_k$ in stage $k$ when agent $1$ manipulates his valuation, we have $g_k' \notin G_{k - 1}$, which implies that $g_k' \in G_{k - 1}' \setminus G_{k - 1} \subseteq X_{k - 1}$.
    Thus, $G_k' \setminus G_k = ((G_{k - 1}' \setminus G_{k - 1}) \cup \{g_k\}) \setminus \{g_k'\}$, which implies that $g_k$ is a new good and $g_k' \notin X_k$.
    Let $M_k(g_k) = M_{k - 1}(g_k')$, and it is easy to verify that Property~\ref{item:property2-for-item-mappings} are satisfied for $M_k$ (Let $M_k(g_k) = M_{k - 1}(g_k')=h$. Then, $g_k'$ was a preimage of $h$ in $M_{k-1}$, and it is replaced by $g_k$ in $M_{k}$. Therefore, the number of preimages of $h$ is unchanged).
    To see that Property~\ref{item:property1-for-item-mappings} holds, note that
    \begin{align*}
        v_1(g_k)
        \leq v_1(b_k)
        = \min_{g \in B_k} v_1(g)
        \leq v_1(M_{k - 1}(g_k'))
        = v_1(M_k(g_k)),
    \end{align*}
    where the last inequality holds by $M_{k - 1}(g_k') \in B_k$.
    As a result, $M_k$ is a valid mapping.
\end{proof}

\subsection{Lower Bound}

In this subsection, we give our improved incentive ratio lower bound for $(\varphi - 1)$-EF1 mechanisms.

\begin{theorem}\label{thm:lb-ef1-subaddi-cancel}
    Let $\varphi = (1 + \sqrt{5}) / 2$.
    Every $(\varphi - 1)$-EF1 mechanism for subadditive cancelable valuations admits an incentive ratio of at least $\varphi$.
\end{theorem}

The proof of \Cref{thm:lb-ef1-subaddi-cancel} is a modification of the proof of \Cref{thm:lb-12-ef1-addi} and is deferred to Appendix~\ref{sec:proof-subaddi-cancel}.
Interestingly, we only need to replace one constructed additive valuation in the proof of \Cref{thm:lb-12-ef1-addi} with a non-additive valuation, and all other valuations remain the same.
\section{Submodular Valuations}

In this section, we consider submodular valuations and show that a generalization of Round-Robin, which satisfies $\frac{1}{2}$-EF1, admits an incentive ratio of $n$.

Given that Mechanism~\ref{alg:round-robin} is not known to possess any fairness property for submodular valuations, we consider a generalization of Round-Robin presented in Mechanism~\ref{alg:round-robin-submodular}.
In particular, instead of receiving an available good with the highest value, agents alternately receive an available good with the highest marginal value with respect to the current bundle.
\cite{DBLP:conf/sigecom/AmanatidisBL0R23} shows that Mechanism~\ref{alg:round-robin-submodular} satisfies $\frac{1}{2}$-EF1 for submodular valuations.
We show in the following theorem that Mechanism~\ref{alg:round-robin-submodular} admits an incentive ratio of $n$ for submodular valuations.

\begin{algorithm}[!htp]
        \caption{Round-Robin for Submodular Valuations}
        \label{alg:round-robin-submodular}
        \begin{algorithmic}[1]
            \State {$S = G$; $(A_1, \ldots, A_n) = (\emptyset, \ldots, \emptyset)$; $k = \lceil m / n \rceil$}
            \For {$r = 1, \ldots, k$}
                \For {$i = 1, \ldots, n$}
                    \State {$g = \arg\max_{h \in S} v_i(h \mid A_i)$} \Comment{Break ties lexicographically.}
                    \State {$A_i = A_i \cup \{g\}$} \Comment{The current agent receives his favorite available good.}
                    \State {$S = S \setminus \{g\}$} \Comment{The good is no longer available.}
                \EndFor
            \EndFor
            \State {\Return $A = (A_1, \ldots, A_n)$}
        \end{algorithmic}
\end{algorithm}

\begin{theorem}\label{thm:incentive-ratio-of-round-robin-for-submodular-valuations}
    Mechanism~\ref{alg:round-robin-submodular} admits an incentive ratio of $n$ for submodular valuations.
\end{theorem}

\begin{proof}
    We prove the upper and lower bounds separately.

    \paragraph{Upper bound.}
    For the upper bound, it suffices to consider agent $1$ since every agent $i$ cannot alter the goods chosen in the first $i - 1$ stages by manipulation.
    Let $A$ be the allocation produced by Mechanism~\ref{alg:round-robin-submodular}.
    We prove a slightly stronger statement that when all agents report truthfully, the utility of agent $1$ constitutes at least a $1 / n$ fraction of his value for $G$, i.e., $v_1(A_1) \geq v_1(G) / n$.
    Given this property, the upper bound holds straightforwardly by the monotonicity of valuations.

    Assume that $m$ is a multiple of $n$ as otherwise, we can achieve this by adding dummy goods with value $0$.
    Thus, Mechanism~\ref{alg:round-robin-submodular} consists of $k = m / n$ rounds.
    We renumber the goods so that $g_i$ is the good received by some agent in stage $i$.
    For every $r \in [k]$, denote $g^r = g_{(r - 1) n + 1}$ as the good received by agent $1$ at round $r$, $L^r = \{g_{(r - 1) n + 1}, g_{(r - 1) n + 2}, \ldots, g_{rn}\}$ as the set of goods received by some agents at round $r$, and $G^r = \{g^1, g^2, \ldots, g^r\}$ as the set of goods received by agent $1$ until the end of round $r$.
    In particular, let $G^0 = \emptyset$.
    By the description of Mechanism~\ref{alg:round-robin-submodular},
    \begin{align}\label{eqn:rr-favor-submodular}
        v_1(g^r \mid G^{r - 1}) = \max_{g \in L^r} v_1(g \mid G^{r - 1})
    \end{align}
    for every $r \in [k]$.
    As a result,
    \begin{align*}
        v_1(G)
        &= \sum_{k=1}^m v_1(g_k \mid \{g_1, \ldots, g_{k - 1}\})
        \leq \sum_{r=1}^k \sum_{g \in L^r} v_1(g \mid G^{r - 1})\\
        &\leq \sum_{r=1}^k n \cdot v_1(g^r \mid G^{r - 1}) 
        = n \cdot v_1(G^k)
        = n \cdot v_1(A_1),
    \end{align*}
    where the first inequality holds by the submodularity of $v_1$ and the second inequality holds by \eqref{eqn:rr-favor-submodular}.
    Therefore, $v_1(A_1) \geq v_1(G) / n$, concluding the proof.

    \paragraph{Lower bound.}
    Let $w$ be a large positive integer, and let $T\geq wn^2$.
    We construct an instance with $n$ agents and $m=wn+T$ goods.
    The set of goods is partitioned by $G_1=\{g_1,\ldots,g_{wn}\}$ and $G_2=\{g_{wn+1},\ldots,g_{wn+T}\}$.
    Let $v_1$ be additive and defined as follows:
    \begin{align*}
        v_1(g) =
        \begin{cases}
            1, & g \in G_1,\\
            0, & g \in G_2.
        \end{cases}
    \end{align*}
    For every $i \in \{2, \ldots, n\}$, to define $v_i$, we first define an additive function $u_i$.
    Let 
    $$C_i=\{g_i,g_{n+i},g_{2n+i},\ldots,g_{(w-1)n+i}\}\subseteq G_1,$$ 
    and $u_i$ is defined as
    \begin{align*}
        u_i(g) =
        \begin{cases}
            w, & g = g_{i - 1}, \\
            2, & g \in C_i, \\
            2, & g \in G_2, \\
            0, & \text{otherwise}.
        \end{cases}
    \end{align*}
    Now, agent $i$'s valuation is defined as
    \begin{align*}
        v_i(S) =
        \begin{cases}
            u_i(S)-|C_i\cap S|, & g_{i-1}\in S, \\
            u_i(S), & \text{otherwise}.
        \end{cases}
    \end{align*}
    To prove that $v_i$ is submodular, we interpret it as a \emph{coverage function}.
    Suppose each $g\in C_i\cup G_2$ corresponds to a set that contains $2$ elements, and every pair of sets in $C_i\cup G_2$ are disjoint.
    Suppose that $g_{i-1}$ corresponds to a set that contains $w$ elements such that $g_{i-1}$ and each $g\in C_i$ intersect at exactly one element.
    It is easy to see that $v_i$ describes the corresponding coverage function and is hence submodular since every coverage function is submodular~\cite{DBLP:books/cu/p/0001G14}.

    If agent $1$ reports $v_1$ truthfully, it is easy to check that for the first $w$ rounds, each agent $i$ receives $C_i$, and this characterizes the allocation of $G_1$.
    As a result, under truthful telling, the utility of agent $1$ is $w$.

    Now, suppose that agent $1$ reports the additive valuation $v_1'$ satisfying
    \begin{align*}
        v_1'(g) =
        \begin{cases}
            1, & g\in G_1\setminus\{g_1,\ldots,g_{n-1}\}, \\
            0, & g\in G_2\cup\{g_1,\ldots,g_{n-1}\}.
        \end{cases}
    \end{align*}
    At the first round, agent $1$ receives $g_n$, and every agent $i \in \{2, \ldots, n\}$ receives $g_{i-1}$.
    At all subsequent rounds, for every $i \in \{2, \ldots, n\}$, the marginal gain of every $g\in C_i$ with respect to $v_i$ is only $1$, and hence, agent $i$ will pick goods from $G_2$.
    Our construction with $T$ being set large enough ensures that agents $2, \ldots, n$ will only pick goods from $G_2$ after the first round.
    As a result, agent $1$ will receive $\{g_n,g_{n+1},\ldots,g_{wn}\}$, which is worth $wn-n+1$ for him.

    Therefore, the incentive ratio of Mechanism~\ref{alg:round-robin-submodular} is lower bounded by $(wn-n+1)/w$, which approaches to $n$ as $w\rightarrow\infty$.
\end{proof}

\section{Discussion and Future Directions}

In this paper, we provide both positive and negative results for the incentive ratio achievable by fair mechanisms for various categories of valuations and leave many open problems.
The most interesting future direction is to close the gaps between the incentive ratio upper and lower bounds.

In addition, we only consider additive, (subadditive) cancelable, and submodular valuations in this paper, while the fair division problem with other valuation classes has also received extensive attention \cite{DBLP:conf/aaai/ChaudhuryGM21,DBLP:conf/ijcai/AkramiRS22}.
Hence, it would be intriguing to investigate broader valuation classes.
In particular, for \emph{fractionally subadditive (XOS)} valuations, which constitute a strict superset of submodular valuations and a strict subset of subadditive valuations, we show in Appendix~\ref{sec:xos} that Mechanism~\ref{alg:round-robin-submodular} admits an incentive ratio of $\lceil m / n \rceil$ and does not provide any fairness guarantee.
Moreover, in Appendix~\ref{sec:envy-graph-procedure}, we show that the \emph{Envy-Graph Procedure} mechanism, which produces EF1 allocations for general valuations \cite{DBLP:conf/sigecom/LiptonMMS04} and, to the best of our knowledge, remains the only known EF1 mechanism even for submodular valuations, admits an infinite incentive ratio for additive valuations.

Finally, the results on divisible resource allocations suggest that allowing randomization usually leads to substantial improvements in incentive guarantees \cite{aziz2014cake,DBLP:journals/corr/abs-2308-08903,mossel2010truthful}.
Hence, it is also natural to study the incentive ratio of randomized fair mechanisms in the indivisible setting.

\bibliographystyle{alpha}
\bibliography{references}

\clearpage
\appendix
\section{Proof of \Cref{thm:lb-ef1-subaddi-cancel}}
\label{sec:proof-subaddi-cancel}

Assume for contradiction that a $(\varphi - 1)$-EF1 mechanism $\calM$ for subadditive cancelable valuations exists with an incentive ratio of $\alpha$ satisfying $1 \leq \alpha < \varphi$.
Suppose that there are $n = 2$ subadditive cancelable agents and $m = 7$ goods.
For every $i \in [2]$, denote $N_i$ as the set of goods $\alpha$-controlled by agent $i$ with respect to $\calM$.
By \Cref{lem:control-lemma}, $(N_1, N_2)$ forms a partition of $G$.
Without loss of generality, assume that $|N_1| \geq 4$ and $\{g_1, g_2, g_3, g_4\} \subseteq N_1$.
Denote $G' = \{g_1, g_2, g_3, g_4\}$, and the value of every constructed subadditive cancelable valuation $v$ in the proof is independent of goods in $G \setminus G'$, i.e., $v(S) = v(S \cap G')$ for every $S \subseteq G$.
For simplicity, we assume goods in $G \setminus G'$ always to be assigned to agent $1$, and we omit them when describing valuations and allocations.

We first define the only non-additive valuation in the proof.
Let $\epsilon > 0$ be an arbitrary real number satisfying $\epsilon < 0.1$ and $\frac{\varphi^2}{\varphi + \epsilon} > \alpha$, and let $v_2$ be a valuation satisfying
\begin{align*}
    v_2(S) =
    \begin{cases}
        0, & |S| = 0,\\
        1, & |S| = 1,\\
        \varphi + \epsilon, & |S| = 2,\\
        \varphi^2, & |S| = 3,\\
        \varphi^2 + \epsilon, & |S| = 4
    \end{cases}
\end{align*}
for every $S \subseteq G'$.
To see that $v_2$ is subadditive, it is easy to verify that $v_2(S \cup T) \leq v_2(S) + v_2(T)$ for all $S, T \subseteq G'$.
Moreover, to see that $v_2$ is cancelable, notice that for all $S, T \subseteq G'$, $v_2(S) > v_2(T)$ iff $|S| > |T|$.
Hence, for all $S, T \subseteq G'$ and $g \in G' \setminus (S \cup T)$, if $v_2(S + g) > v_2(T + g)$, then $|S \cup \{g\}| > |T \cup \{g\}|$, which implies that $|S| > |T|$ and thereby $v_2(S) > v_2(T)$.

We emphasize again that all valuations in the proof except for $v_2$ are additive, and the first profile $\bdv^{(0)} = (v_1, v_2)$ is defined as

\begin{table}[H]
\centering
\begin{tabular}{@{}cccccc@{}}
\toprule
\multicolumn{1}{l}{}          &       & $g_1$ & $g_2$ & $g_3$ & $g_4$ \\ \midrule
\multirow{2}{*}{$\bdv^{(0)}$} & $v_1$ & $1$   & $1$   & $1$   & $1$   \\
                              & $v_2$ &  &   &   &   \\ \bottomrule
\end{tabular}
\end{table}

\noindent By the $(\varphi - 1)$-EF1 property of $\calM$, $|\calM_1(\bdv^{(0)})| = |\calM_2(\bdv^{(0)})| = 2$.
Without loss of generality, assume that $\calM_1(\bdv^{(0)}) = \{g_1, g_2\}$ and $\calM_2(\bdv^{(0)}) = \{g_3, g_4\}$.

Let $\delta$ be an arbitrary real number with $\delta \geq 5$, and we consider the next profile $\bdv^{(1)} = (v_1', v_2)$ where

\begin{table}[H]
\centering
\begin{tabular}{@{}cccccc@{}}
\toprule
\multicolumn{1}{l}{}          &       & $g_1$ & $g_2$ & $g_3$ & $g_4$ \\ \midrule
\multirow{2}{*}{$\bdv^{(1)}$} & $v_1'$ & $\delta$   & $\varphi \delta$   & $0$   & $0$   \\
                              & $v_2$ &    &    &    &    \\ \bottomrule
\end{tabular}
\end{table}

\noindent We claim that $\calM_1(\bdv^{(1)}) = \{g_1, g_2\}$ and $\calM_2(\bdv^{(1)}) = \{g_3, g_4\}$.
Firstly, by the $(\varphi - 1)$-EF1 property of $\calM$, $|\calM_2(\bdv^{(1)})| \geq 2$.
Moreover, if $|\{g_1, g_2\} \cap \calM_1(\bdv^{(1)})| < 2$, by deviating from $\bdv^{(1)}$ to $\bdv^{(0)}$, agent $1$ can increase his utility in $\bdv^{(1)}$ by a factor of
\begin{align*}
    \frac{v_1'(\calM_1(\bdv^{(0)}))}{v_1'(\calM_1(\bdv^{(1)}))}
    \geq \frac{v_1'(\{g_1, g_2\})}{v_1'(g_2)}
    = \frac{\varphi + 1}{\varphi}
    = \varphi
    > \alpha,
\end{align*}
violating the incentive ratio $\alpha$ of $\calM$.
Hence, $\{g_1, g_2\} \subseteq \calM_1(\bdv^{(1)})$, and it follows that $\calM_1(\bdv^{(1)}) = \{g_1, g_2\}$ and $\calM_2(\bdv^{(1)}) = \{g_3, g_4\}$.

We proceed to the next profile $\bdv^{(2)} = (v_1'', v_2)$ where

\begin{table}[H]
\centering
\begin{tabular}{@{}cccccc@{}}
\toprule
\multicolumn{1}{l}{}          &       & $g_1$ & $g_2$ & $g_3$ & $g_4$ \\ \midrule
\multirow{2}{*}{$\bdv^{(2)}$} & $v_1''$ & $\varphi \delta$   & $\delta$   & $0$   & $0$   \\
                              & $v_2$ &    &    &    &    \\ \bottomrule
\end{tabular}
\end{table}

\noindent Analogous to $\bdv^{(1)}$, we can show that $\calM_1(\bdv^{(2)}) = \{g_1, g_2\}$ and $\calM_2(\bdv^{(2)}) = \{g_3, g_4\}$.

In the next profile, we modify the valuation of agent $2$.
Define $\bdv^{(3)} = (v_1', v_2')$ where

\begin{table}[H]
\centering
\begin{tabular}{@{}cccccc@{}}
\toprule
\multicolumn{1}{l}{}          &       & $g_1$ & $g_2$ & $g_3$ & $g_4$ \\ \midrule
\multirow{2}{*}{$\bdv^{(3)}$} & $v_1'$ & $\delta$   & $\varphi \delta$   & $0$   & $0$   \\
                              & $v_2'$ & $0$   & $\delta$   & $1$   & $1$   \\ \bottomrule
\end{tabular}
\end{table}

\noindent We claim that $\calM_1(\bdv^{(3)}) = \{g_1, g_2\}$ and $\calM_2(\bdv^{(3)}) = \{g_3, g_4\}$.
Firstly, since agent $1$ $\alpha$-strongly desires $g_2$ in $\bdv^{(3)}$, $g_2 \in \calM_1(\bdv^{(3)})$ by the assumption that agent $1$ $\alpha$-controls $g_2$ with respect to $\calM$.
Moreover, if $|\{g_3, g_4\} \cap \calM_2(\bdv^{(3)})| < 2$, by deviating from $\bdv^{(3)}$ to $\bdv^{(1)}$, agent $2$ can increase his utility in $\bdv^{(3)}$ by a factor of
\begin{align*}
    \frac{v_2'(\calM_2(\bdv^{(1)}))}{v_2'(\calM_2(\bdv^{(3)}))}
    \geq 2
    > \alpha,
\end{align*}
violating the incentive ratio $\alpha$ of $\calM$.
Hence, $\{g_3, g_4\} \subseteq \calM_2(\bdv^{(3)})$.
Finally, if $|\calM_2(\bdv^{(3)})| > 2$, by deviating from $\bdv^{(1)}$ to $\bdv^{(3)}$, agent $2$ can increase his utility in $\bdv^{(1)}$ by a factor of
\begin{align*}
    \frac{v_2(\calM_2(\bdv^{(3)}))}{v_2(\calM_2(\bdv^{(1)}))}
    \geq \frac{\varphi^2}{\varphi + \epsilon}
    > \alpha,
\end{align*}
where the last inequality holds by the definition of $\epsilon$, violating the incentive ratio $\alpha$ of $\calM$.
As a result, $|\calM_2(\bdv^{(3)})| \leq 2$, and it follows that $\calM_1(\bdv^{(3)}) = \{g_1, g_2\}$ and $\calM_2(\bdv^{(3)}) = \{g_3, g_4\}$.

In the next profile, we manage to allocate $g_2$ to agent $2$.
Define $\bdv^{(4)} = (v_1'', v_2'')$ where

\begin{table}[H]
\centering
\begin{tabular}{@{}cccccc@{}}
\toprule
\multicolumn{1}{l}{}          &       & $g_1$ & $g_2$ & $g_3$ & $g_4$ \\ \midrule
\multirow{2}{*}{$\bdv^{(4)}$} & $v_1''$ & $\varphi \delta$   & $\delta$   & $0$   & $0$   \\
                              & $v_2''$ & $\delta$   & $\delta$   & $1$   & $1$   \\ \bottomrule
\end{tabular}
\end{table}

\noindent We claim that $g_1 \in \calM_1(\bdv^{(4)})$ and $g_2 \in \calM_2(\bdv^{(4)})$.
Firstly, since agent $1$ $\alpha$-strongly desires $g_1$ in $\bdv^{(4)}$, $g_1 \in \calM_1(\bdv^{(4)})$ by the assumption that agent $1$ $\alpha$-controls $g_1$ with respect to $\calM$.
Moreover, if $\{g_1, g_2\} \subseteq \calM_1(\bdv^{(4)})$, then $\calM(\bdv^{(4)})$ is not $(\varphi - 1)$-EF1 for agent $2$, violating the $(\varphi - 1)$-EF1 property of $\calM$.
Hence, $|\{g_1, g_2\} \cap \calM_1(\bdv^{(4)})| \leq 1$, and it follows that $g_1 \in \calM_1(\bdv^{(4)})$ and $g_2 \in \calM_2(\bdv^{(4)})$.

We present our final profile to derive a contradiction.
Define $\bdv^{(5)} = (v_1'', v_2')$ where

\begin{table}[H]
\centering
\begin{tabular}{@{}cccccc@{}}
\toprule
\multicolumn{1}{l}{}          &       & $g_1$ & $g_2$ & $g_3$ & $g_4$ \\ \midrule
\multirow{2}{*}{$\bdv^{(5)}$} & $v_1''$ & $\varphi \delta$   & $\delta$   & $0$   & $0$   \\
                              & $v_2'$ & $0$   & $\delta$   & $1$   & $1$   \\ \bottomrule
\end{tabular}
\end{table}

\noindent Firstly, if $|\{g_1, g_2\} \cap \calM_1(\bdv^{(5)})| < 2$, by deviating from $\bdv^{(5)}$ to $\bdv^{(3)}$, agent $1$ can increase his utility in $\bdv^{(5)}$ by a factor of
\begin{align*}
    \frac{v_1''(\calM_1(\bdv^{(3)}))}{v_1''(\calM_1(\bdv^{(5)}))}
    \geq \frac{v_1''(\{g_1, g_2\})}{v_1''(\{g_1\})}
    = \frac{\varphi + 1}{\varphi}
    = \varphi
    > \alpha,
\end{align*}
violating the incentive ratio $\alpha$ of $\calM$.
Hence, $\{g_1, g_2\} \subseteq \calM_1(\bdv^{(5)})$.
However, by deviating from $\bdv^{(5)}$ to $\bdv^{(4)}$, agent $2$ can increase his utility in $\bdv^{(5)}$ by a factor of
\begin{align*}
    \frac{v_2'(\calM_2(\bdv^{(4)}))}{v_2'(\calM_2(\bdv^{(5)}))}
    \geq \frac{v_2'(g_2)}{v_2'(\{g_3, g_4\})}
    = \frac{\delta}{2}
    > \alpha,
\end{align*}
violating the incentive ratio $\alpha$ of $\calM$.
This concludes the proof of \Cref{thm:lb-ef1-subaddi-cancel}.

\section{XOS Valuations}
\label{sec:xos}

In this section, we consider \emph{fractionally subadditive (XOS)} valuations and show that for $n \geq 2$ and $m \geq n$, the incentive ratio of Mechanism~\ref{alg:round-robin-submodular} is $\lceil m / n \rceil$.
Recall that a valuation $v$ is XOS if there exists a finite set of additive functions $\{f_1, \ldots, f_{\alpha}\}$ such that $v(S) = \max_{k \in [\alpha]} f_k(S)$ for every $S \subseteq G$.
Our analysis also implies that for all $n \geq 2$ and $m \geq n$, there exists an instance with $n$ agents and $m$ goods such that the marginal value of each good lies between $[0, 1]$, and the allocation $A$ produced by Mechanism~\ref{alg:round-robin-submodular} admits a maximum envy of $\Theta(m / n)$, i.e., $\max_{i \neq j} (v_i(A_j) - v_i(A_i)) = \Theta(m / n)$.
Consequently, Mechanism~\ref{alg:round-robin-submodular} does not provide any meaningful fairness guarantee for XOS valuations as $m$ tends to infinity.

\begin{theorem}\label{thm:RR-XOS}
    For $n \geq 2$ and $m \geq n$, Mechanism~\ref{alg:round-robin-submodular} admits an incentive ratio of $\lceil m / n \rceil$ for XOS valuations.
\end{theorem}

\begin{proof}
    We prove the upper and lower bounds separately.

    \paragraph{Upper bound.}
    We prove a stronger statement that the incentive ratio for each agent cannot exceed the number of goods he receives, which is at most $\lceil m / n \rceil$.
    In particular, we prove this statement for agent $1$, and the statement for agent $i \in \{2, \ldots, n\}$ can be reduced to that for agent $1$ by noticing that agent $i$ cannot alter the outcomes in the first $i - 1$ stages by manipulation.
    Let $v_1(S) = \max_{k \in [\alpha]} f_k(S)$ for every $S \subseteq G$, where $f_1, \ldots, f_{\alpha}$ are additive functions.
    We first show that for all $S \subseteq G$ and $g \in G \setminus S$,
    \begin{align}\label{eqn:margin-mono}
        v_1(g \mid S) \leq v_1(g).
    \end{align}
    This is because
    \begin{align*}
        v_1(g \mid S)
        &= v_1(S + g) - v_1(S)
        = \max_{k \in [\alpha]} f_k(S + g) - \max_{k \in [\alpha]} f_k(S)\\
        &\leq \max_{k \in [\alpha]} (f_k(S + g) - f_k(S))
        = \max_{k \in [\alpha]} f_k(g)
        = v_1(g).
    \end{align*}

    Notice that the number of goods received by agent $1$, despite his reported valuation, is exactly $s := \lceil m / n \rceil$, and denote $G' = \{g_1', \ldots, g_s'\}$ and $G'' = \{g_1'', \ldots, g_s''\}$ as the sets of goods received by agent $1$ when he reports truthfully and manipulates, respectively, where $g_i'$ and $g_i''$ are the goods allocated to him at round $i$.
    It suffices to show that $v_1(G'') \leq s \cdot v_1(G')$.
    By the description of Mechanism~\ref{alg:round-robin-submodular},
    \begin{align}\label{eqn:max-first-good}
        v_1(g_1') = \max_{g \in G} v_1(g).
    \end{align}
    As a result,
    \begin{align*}
        v_1(G'')
        = \sum_{k=1}^s v_1(g_k'' \mid \{g_1'', \ldots, g_{k - 1}''\})
        \leq \sum_{k=1}^s v_1(g_k'')
        \leq \sum_{k=1}^s v_1(g_1')
        = s \cdot v_1(g_1')
        \leq s \cdot v_1(G'),
    \end{align*}
    where the first inequality holds by \eqref{eqn:margin-mono}, the second inequality holds by \eqref{eqn:max-first-good}, and the third inequality holds by the monotonicity of $v_1$.

    \paragraph{Lower bound.}
    For every $i \in N$, let $s_i := \lceil (m-i+1) / n \rceil$ denote the number of goods received by agent $i$, and it holds that $\sum_{i \in N} s_i = m$.
    Note that $|s_i - s_j| \leq 1$ for all $i, j \in N$.
    Partition $G$ into $n$ groups $G_1, \ldots, G_n$ such that $G_1 = \{g_1, \ldots, g_{s_1}\}$, $G_2 = \{g_{s_1 +1}, \ldots, g_{s_1 + s_2}\}$, and so on.
    We will construct an instance such that every agent $i \in N$ receives bundle $G_i$ when all agents report truthfully.
    Then, we show that by misreporting, agent $1$ can obtain the entire $G_2$ and, if $s_1 = s_2 + 1$, an additional good of $g_{s_1}$.
    Now, we formally describe our hard instance.
    For every $S \subseteq G$, let $v_1(S) = \max\{f_1^1(S), f_2^1(S)\}$ where additive functions $f_1^1, f_2^1$ are defined as
    \begin{align*}
        f_1^1(g) = 
        \begin{cases}
            1, & g = g_1,\\
            0, & g \in G \setminus \{g_1\},
        \end{cases}
        \qquad \text{and} \qquad
        f_2^1(g) =
        \begin{cases}
            0, & g \in G \setminus G_2 \setminus \{g_{s_1}\},\\
            1, & g \in G_2 \cup \{g_{s_1}\},
        \end{cases}
    \end{align*}
    and $v_2(S) = \max\{f_1^2(S), f_2^2(S)\}$ where additive functions $f_1^2, f_2^2$ are defined as
    \begin{align*}
        f_1^2(g) = 
        \begin{cases}
            1, & g \in G_1,\\
            0, & g \in G \setminus G_1,
        \end{cases}
        \qquad \text{and} \qquad
        f_2^2(g) =
        \begin{cases}
            0, & g \in G \setminus G_2,\\
            2, & g = g_{s_1 + 1},\\
            1, & g \in G_2 \setminus \{g_{s_1 + 1}\}.
        \end{cases}
    \end{align*}
    For every agent $i \in \{3, \ldots, n\}$, let $v_i$ be an additive function satisfying
    \begin{align*}
        v_i(g) =
        \begin{cases}
            0, & g \in G \setminus G_i,\\
            1, & g \in G_i.
        \end{cases}
    \end{align*}
    
    On one hand, assume that all agents report truthfully.
    At the first round, agent $1$ receives $g_1$, agent $2$ receives $g_{s_1 + 1}$, and every agent $i \in \{3, \ldots, n\}$ receives a good in $G_i$.
    At the subsequent rounds, the marginal value of each remaining good is $0$ for agent $1$, and hence, agent $1$ prefers goods in $G_1$ to all other remaining goods due to the lexicographic tie-breaking rule; the marginal value for agent $2$ is $1$ for each remaining good in $G_2$ and is $0$ for all other remaining goods, and hence, agent $2$ prefers goods in $G_2$; every agent $i \in \{3, \ldots, n\}$ prefers goods in $G_i$.
    Consequently, the resulting allocation is $(G_1, \ldots, G_n)$ and the utility of agent $1$ is $v_1(G_1) = 1$.
    
    On the other hand, assume that agent $1$ manipulates his valuation as $v_1'$ where $v_1'$ is an additive function satisfying
    \begin{align*}
        v_1'(g) = 
        \begin{cases}
            0, & g \in G \setminus \{G_2\} \setminus \{g_{s_1}\},\\
            1, & g = g_{s_1}, \\
            2, & g \in G_2.
        \end{cases}
    \end{align*}
    Note that agent $1$ favors goods in $G_2$ the most and $g_{s_1}$ the second.
    At the first round, agent $1$ receives $g_{s_1 + 1}$, agent $2$ receives $g_1$, and every agent $i \in \{3, \ldots, n\}$ receives a good in $G_i$.
    At the subsequent rounds, the marginal value for agent $2$ is $1$ for each remaining good in $G_1$ and is $0$ for all other remaining goods, and hence, agent $2$ prefers goods in $G_1$; every agent $i \in \{3, \ldots, n\}$ prefers goods in $G_i$.
    Finally, at the last round, if only one good is left, which must be $g_{s_1}$ by the tie-breaking rule, then it will be received by agent $1$.
    As a result, if $s_1 = s_2$, then the resulting allocation is $(G_2, G_1, G_3, G_4, \ldots, G_n)$; otherwise, the resulting allocation is $(G_2 \cup \{g_{s_1}\}, G_1 \setminus \{g_{s_1}\}, G_3, G_4, \ldots, G_n)$.
    In both cases, the utility of agent $1$ is $s_1 = \lceil m / n \rceil$.
    Therefore, the incentive ratio of Mechanism~\ref{alg:round-robin-submodular} is lower bounded by $\lceil m / n \rceil$.
\end{proof}

As a corollary of the proof of \Cref{thm:RR-XOS}, we show that Mechanism~\ref{alg:round-robin-submodular} does not provide any meaningful fairness guarantees for XOS valuations.

\begin{corollary}\label{cor:rr-xos}
    Assume that all marginal values of goods lie between $[0, 1]$.
    For all $n \geq 2$ and $m \geq n$, an instance with $n$ XOS agents and $m$ goods exists such that the allocation $A$ produced by Mechanism~\ref{alg:round-robin-submodular} admits a maximum envy of $\Theta(m / n)$, i.e., $\max_{i \neq j} (v_i(A_j) - v_i(A_i)) = \Theta(m / n)$.
\end{corollary}

\begin{proof}
Note that in the hard instance given in the proof of the lower bound in \Cref{thm:RR-XOS}, all marginal values of goods for agent $1$ lie between $[0, 1]$, and when all agents report truthfully, the allocation $(G_1, \ldots, G_n)$ returned by Mechanism~\ref{alg:round-robin-submodular} satisfies $v_1(G_2) - v_1(G_1) = \Theta(m / n)$.
\end{proof}
\section{Envy-Graph Procedure}
\label{sec:envy-graph-procedure}

In this section, we adopt two implementations of the \emph{Envy-Graph Procedure} mechanism and show that both of them admit an infinite incentive ratio for additive valuations.
To describe Envy-Graph Procedure, we first define the notion of \emph{envy graphs}.
The envy graph of an allocation $A = (A_1, \ldots, A_n)$ includes a vertex for each agent, and a directed edge from $i$ to $j$ exists iff agent $i$ envies agent $j$, i.e., $v_i(A_j) > v_i(A_i)$.
We present the first implementation of Envy-Graph Procedure in Mechanism~\ref{alg:envy-cycle-elimination}, which enumerates all goods in $G$ according to a pre-specified order and ensures that the envy graph is always acyclic.

In each iteration, we first find an unenvied agent $j$, i.e., a source vertex in the envy graph, with respect to the current allocation $A$ (Line~\ref{code:find-unenvied-agent}) and give the good to agent $j$ (Line~\ref{code:give-the-good-to-the-unenvied-agent}).
Then, we eliminate all cycles in the envy graph (Line~\ref{code:eliminate-envy-cycles}).
Specifically, whenever a cycle exists in the envy graph, supposing to be $1 \to 2 \to \ldots \to c \to 1$ without loss of generality, we derive a new allocation $A'$ where $A_i' = A_{(i \bmod c) + 1}$ for every $i \in [c]$ and $A_i' = A_i$ for every $i \in \{c + 1, \ldots, n\}$, and replace $A$ with $A'$.
This elimination process terminates after at most $O(n^2)$ steps as the number of edges strictly decreases each time (see, e.g., \cite[Theorem 6.2]{DBLP:journals/siamdm/PlautR20}).
We show in the following theorem that Mechanism~\ref{alg:envy-cycle-elimination} admits an infinite incentive ratio for additive valuations.

\begin{algorithm}[!htp]
        \caption{Envy-Graph Procedure}
        \label{alg:envy-cycle-elimination}
        \begin{algorithmic}[1]
            \State {$(A_1, \ldots, A_n) \gets (\emptyset, \ldots, \emptyset)$}
            \For {$i = 1, \ldots, m$}
                \State \label{code:find-unenvied-agent} {$j \gets \Call{FindUnenviedAgent}{A_1, \ldots, A_n}$} \Comment{Break ties lexicographically.}
                \State \label{code:give-the-good-to-the-unenvied-agent} {$A_j \gets A_j \cup \{g_i\}$}
                \State \label{code:eliminate-envy-cycles} {$(A_1, \ldots, A_n) \gets \Call{EliminateEnvyCycles}{A_1, \ldots, A_n}$}
            \EndFor
            \State {\Return $A = (A_1, \ldots, A_n)$}
        \end{algorithmic}
\end{algorithm}

\begin{theorem}\label{thm:incentive-ratio-of-envy-cycle-elimination-version1}
    Mechanism~\ref{alg:envy-cycle-elimination} admits an infinite incentive ratio for additive valuations.
\end{theorem}

\begin{proof}
    Fix $0 < \epsilon < 1$.
    Suppose that there are $n = 2$ agents and $m = 3$ goods with additive valuation profile $v_1 = [0, 0, 0]$ and $v_2 = [\epsilon, \epsilon, 1]$.
    When both agents report truthfully, we demonstrate the intermediate statuses in the execution of Mechanism~\ref{alg:envy-cycle-elimination} with profile $(v_1, v_2)$:
    \begin{itemize}
        \item Initially, $A = (\emptyset, \emptyset)$, and no edge exists in the envy graph.
        \item \textbf{Iteration $i = 1$:} $j = 1$.
        $A = (\{g_1\}, \emptyset)$, and the envy graph contains edge $2 \to 1$.
        \item \textbf{Iteration $i = 2$:} $j = 2$.
        $A = (\{g_1\}, \{g_2\})$, and no edge exists in the envy graph.
        \item \textbf{Iteration $i = 3$:} $j = 1$.
        $A = (\{g_1, g_3\}, \{g_2\})$, and the envy graph contains edge $2 \to 1$.
    \end{itemize}
    Thus, the resulting allocation is $(\{g_1, g_3\}, \{g_2\})$, and the utility of agent $2$ is $\epsilon$.
    
    Suppose that agent $2$ manipulates his valuation as $v_2' = [1, 0, 0]$.
    We demonstrate the intermediate statuses in the execution of Mechanism~\ref{alg:envy-cycle-elimination} with profile $(v_1, v_2')$:
    \begin{itemize}
        \item Initially, $A = (\emptyset, \emptyset)$, and no edge exists in the envy graph.
        \item \textbf{Iteration $i = 1$:} $j = 1$.
        $A = (\{g_1\}, \emptyset)$, and the envy graph contains edge $2 \to 1$.
        \item \textbf{Iteration $i = 2$:} $j = 2$.
        $A = (\{g_1\}, \{g_2\})$, and the envy graph contains edge $2 \to 1$.
        \item \textbf{Iteration $i = 3$:} $j = 2$.
        $A = (\{g_1\}, \{g_2, g_3\})$, and the envy graph contains edge $2 \to 1$.
    \end{itemize}
    Thus, the resulting allocation is $(\{g_1\}, \{g_2, g_3\})$, and the utility of agent $2$ with respect to $v_2$ is $1 + \epsilon$.
    Therefore, the incentive ratio of Mechanism~\ref{alg:envy-cycle-elimination} is lower bounded by $(1 + \epsilon) / \epsilon$.
    Since $\epsilon$ can be arbitrarily small, we conclude the proof.
\end{proof}

\subsection{Another Implementation}

One may suggest a seemingly more efficient implementation of Envy-Graph Procedure.
As presented in Mechanism~\ref{alg:envy-cycle-elimination-with-another-implementation}, instead of specifying an order for inserted goods, the source agent receives his favorite good among the remaining goods in each stage.
This implementation not only preserves the EF1 property, but also produces EFX allocations for identical ordinal preferences, i.e., for all agents $i$ and $j$, and for all goods $g_1$ and $g_2$, $v_i(g_1) \geq v_i(g_2)$ whenever $v_j(g_1) \geq v_j(g_2)$ \cite{DBLP:journals/siamdm/PlautR20}.
We show that such an implementation also admits an infinite incentive ratio for additive valuations.

\begin{algorithm}[!htp]
        \caption{Another Implementation of Envy-Graph Procedure}
        \label{alg:envy-cycle-elimination-with-another-implementation}
        \begin{algorithmic}[1]
            \State {$S \gets G$; $(A_1, \ldots, A_n) \gets (\emptyset, \ldots, \emptyset)$}
            \For {$i = 1, \ldots, m$}
                \State {$j \gets \Call{FindUnenviedAgent}{A_1, \ldots, A_n}$} \Comment{Break ties lexicographically.}
                \State {$g \gets \arg\max_{h \in S} v_j(h)$} \Comment{Break ties lexicographically.}
                \State {$A_j \gets A_j \cup \{g\}$}
                \State {$(A_1, \ldots, A_n) \gets \Call{EliminateEnvyCycles}{A_1, \ldots, A_n}$}
            \EndFor
            \State {\Return $A = (A_1, \ldots, A_n)$}
        \end{algorithmic}
\end{algorithm}

\begin{theorem}\label{thm:incentive-ratio-of-envy-cycle-elimination-version2}
    Mechanism~\ref{alg:envy-cycle-elimination-with-another-implementation} admits an infinite incentive ratio for additive valuations.
\end{theorem}

\begin{proof}
    Fix $0 < \epsilon < 1$.
    Suppose that there are $3$ agents and $4$ goods with additive valuation profile $v_1 = [1, 0.6, 0, 0.6]$, $v_2 = [1, 0, 0, \epsilon]$, and $v_3 = [0, 1, 1, 1]$.
    When all agents report truthfully, we demonstrate the intermediate statuses in the execution of Mechanism~\ref{alg:envy-cycle-elimination-with-another-implementation} with profile $(v_1, v_2, v_3)$:
    \begin{enumerate}
        \item Initially, $A = (\emptyset, \emptyset, \emptyset)$, and no edge exists in the envy graph.
        \item \textbf{Iteration $i = 1$:} $j = 1$ and $g = g_1$.
        $A = (\{g_1\}, \emptyset, \emptyset)$, and the envy graph contains edge $2 \to 1$.
        \item \textbf{Iteration $i = 2$:} $j = 2$ and $g = g_4$.
        $A = (\{g_1\}, \{g_4\}, \emptyset)$, and the envy graph contains edges $2 \to 1$ and $3 \to 2$.
        \item \textbf{Iteration $i = 3$:} $j = 3$ and $g = g_2$.
        $A = (\{g_1\}, \{g_4\}, \{g_2\})$, and the envy graph contains edge $2\to 1$.
        \item \textbf{Iteration $i = 4$:} $j = 2$ and $g = g_3$.
        $A = (\{g_1\}, \{g_3, g_4\}, \{g_2\})$, and the envy graph contains edges $2\to 1$ and $3 \to 2$.
    \end{enumerate}
    Thus, the resulting allocation is $(\{g_1\}, \{g_3, g_4\}, \{g_2\})$, and the utility of agent $2$ is $\epsilon$.

    Suppose that agent $2$ manipulates his valuation as $v_2' = [1, \epsilon, 0, 0]$.
    We demonstrate the intermediate statuses in the execution of Mechanism~\ref{alg:envy-cycle-elimination-with-another-implementation} with profile $(v_1, v_2', v_3)$:
    \begin{enumerate}
        \item Initially, $A = (\emptyset, \emptyset, \emptyset)$, and no edge exists in the envy graph.
        \item \textbf{Iteration $i = 1$:} $j = 1$ and $g = g_1$.
        $A = (\{g_1\}, \emptyset, \emptyset)$, and the envy graph contains edge $2 \to 1$.
        \item \textbf{Iteration $i = 2$:} $j = 2$ and $g = g_2$.
        $A = (\{g_1\}, \{g_2\}, \emptyset)$, and the envy graph contains edges $2 \to 1$ and $3 \to 2$.
        \item \textbf{Iteration $i = 3$:} $j = 3$ and $g = g_3$.
        $A = (\{g_1\}, \{g_2\}, \{g_3\})$, and the envy graph contains edge $2\to 1$.
        \item \textbf{Iteration $i = 4$:} $j = 2$ and $g = g_4$.
        $A = (\{g_1\}, \{g_2, g_4\}, \{g_3\})$, and the envy graph contains edges $2\to 1$, $1 \to 2$, and $3 \to 2$.
        Due to the existence of cycle $1 \to 2 \to 1$, bundles $A_1$ and $A_2$ are swapped.
        After that, the allocation becomes $A = (\{g_2, g_4\}, \{g_1\}, \{g_3\})$, and the envy graph contains edge $3 \to 1$.
    \end{enumerate}
    Thus, the resulting allocation is $(\{g_2, g_4\}, \{g_1\}, \{g_3\})$, and the utility of agent $2$ with respect to $v_2$ is $1$.
    Therefore, the incentive ratio of Mechanism~\ref{alg:envy-cycle-elimination-with-another-implementation} is lower bounded by $1 / \epsilon$.
    Since $\epsilon$ can be arbitrarily small, we conclude the proof.
\end{proof}

\end{document}